\documentclass[conference]{IEEEtran}

\usepackage{amssymb,amsmath,amsthm, amscd,ifthen}
\usepackage{graphicx}
\usepackage[ruled,vlined]{algorithm2e}
\usepackage {color}
\usepackage{todonotes}

\newtheorem{theorem} {Theorem}
\newtheorem{prop} {Proposition}
\newtheorem{Example}{Example}

\renewcommand{\le}{\leqslant}
\renewcommand{\ge}{\geqslant}
\newcommand{\E}{\mathbb{E}}
\newcommand{\Prob}{\mathrm{P}}

\DeclareMathOperator*{\maximize}{\rm maximize\,}

\begin{document}
\title{Quick Detection of High-degree Entities \\ in Large Directed Networks}

\author{K. Avrachenkov \\ Inria \\ k.avrachenkov@inria.fr \and N. Litvak \\ University of Twente \\ n.litvak@utwente.nl \and
L. Ostroumova Prokhorenkova \\ Yandex \\ ostroumova-la@yandex.ru \and E. Suyargulova \\ Yandex \\ siyargul@yandex.ua}


%

\maketitle

\begin{abstract}
\boldmath
In this paper we address the problem of quick detection of high-degree entities in large online social networks. Practical importance of this problem is attested by a large number of companies that continuously collect and update statistics about popular entities, usually using the degree of an entity as an approximation of its popularity. We suggest a simple, efficient, and easy to implement two-stage randomized algorithm that provides highly accurate solutions to this problem. For instance, our algorithm needs only one thousand API requests in order to find the top-100 most followed users, with more than 90\% precision, in the online social network Twitter with approximately a billion of registered users. Our  algorithm significantly outperforms existing methods and serves many different purposes such as finding the most popular users or the most popular interest groups in social networks. An important contribution of this work is the analysis of the proposed algorithm using Extreme Value Theory~--- a branch of probability that studies extreme events and properties of largest order statistics in random samples. Using this theory we derive an accurate prediction for the algorithm's performance and show that the number of API requests for finding the top-$k$ most popular entities is sublinear in the number of entities. Moreover, we formally show that the high variability of the entities, expressed through heavy-tailed distributions, is the reason for the algorithm's efficiency. We quantify this phenomenon in a rigorous mathematical way.

\footnotetext{The authors are given in alphabetical order. L. Ostroumova Prokhorenkova is the principal author.}

\end{abstract}

\IEEEpeerreviewmaketitle

\section{Introduction}\label{sec:introduction}

In this paper we propose a randomized algorithm for quick detection of high-degree entities in large online social networks. The entities can be, for example, users, interest groups, user categories, geographical locations, etc. For instance, one can be interested in finding a list of Twitter users with many followers or Facebook interest groups with many members. The importance of this problem is attested by a large number of companies that continuously collect and  update statistics about popular entities in online social networks ({\it twittercounter.com}, {\it followerwonk.com}, {\it twitaholic.com}, {\it www.insidefacebook.com}, {\it yavkontakte.ru} just to name a few).

The problem under consideration may seem trivial if one assumes that the network structure and the relation between entities are known. However, even then finding for example the top-$k$ in-degree nodes in a directed graph $G$ of size $N$ takes the time ${\rm O}(N)$. For very large networks, even linear complexity is too high cost to pay. Furthermore, the data of current social networks is typically available only to managers of social networks and can be obtained by other interested parties only through API (Application Programming Interface) requests.
API is a set of request messages, along with a definition of the structure of response messages.
Using one API request it is usually possible to discover either friends of one given user, or his/her interest groups, or the date when his/her account was created, etc. The rate of allowed API requests is usually very limited. For instance, Twitter has the limit of one access per minute for one standard API account (see \textit{dev.twitter.com}). Then, in order to crawl the entire network with a billion users, using one standard API account, one needs more than 1900 years.

Hence currently, there is a rapidly growing interest in algorithms that
 evaluate specific network properties, using only local information (e.g., the degree of a node and its neighbors), and give a good approximate answer in the number of steps that is sublinear in the network size.
Recently, such algorithms have been proposed for PageRank evaluation~\cite{Avrachenkov2011Top-kPPR,SublinearPageRank,Borgs2013PageRank}, for finding high-degree nodes in graphs~\cite{Avrachenkov2012Top-k,Brautbar2010high_degree,Cooper2012high_degree,Kumar2008}, and for finding the root of a preferential attachment tree~\cite{Borgs2012pa_root}.

In this paper, we propose a new two-stage method for finding high-degree nodes in large directed networks with highly skewed in-degree distribution. We demonstrate that our algorithm outperforms other known methods by a large margin and has a better precision than the for-profit Twitter statistics {\it twittercounter.com}.

\section{Problem formulation and our contribution}

Let $V$ be a set of $N$ entities, typically users, that can be accessed using API requests. Let $W$ be another set of $M$ entities (possibly equal to $V$).  We consider a bipartite graph $(V,W,E)$, where a directed edge $(v,w)\in E$, with $v\in V$, and $w\in W$, represents a relation between $v$ and $w$. In our particular model of the Twitter graph
$V$ is a set of Twitter users, $W=V$, and $(v,w)\in E$ means that $v$ follows $w$ or that $v$ retweeted a tweet of $w$. Note that any directed graph $G=(V,E)$ can be represented equivalently by the bipartite graph $(V,V,E)$. One can also suppose that $V$ is a set of users and $W$ is a set of interest groups, while the edge $(v,w)$ represents that the user $v$ belongs to the group $w$.

Our goal is to quickly find the top in-degree entities in $W$. In this setting, throughout the paper, we use the terms `nodes', `vertices',  and `entities' interchangeably.

We propose a very simple and easy-to-implement algorithm that detects popular entities with high precision using a surprisingly small number of API requests. Most of our experiments are performed on the Twitter graph, because it is a good example of a huge network (approximately a billion of registered users) with a very limited rate of requests to API. We use only 1000 API requests to find the top-100 Twitter users with a very high precision.  We also demonstrate the efficacy of our approach on the popular Russian online social network VKontakte (\textit{vk.com}) with more than 200 million registered users. We use our algorithm to quickly detect the most popular interest groups in this social network. Our experimental analysis shows that despite of its simplicity, our algorithm significantly outperforms existing approaches, e.g., \cite{Avrachenkov2012Top-k,Brautbar2010high_degree,Kumar2008}. Moreover, our algorithm can be used in a very general setting for finding the most popular entities, while some baseline algorithms can only be used for finding nodes of largest degrees in directed \cite{Kumar2008} or undirected \cite{Avrachenkov2012Top-k} graphs.

In most social networks the degrees of entities show great variability.  This is often modeled using power laws, although it has been often argued that the classical Pareto distribution does not always fit the observed data. In our analysis we assume that the incoming degrees of the entities in $W$ are independent random variables following a {\it regularly varying} distribution $G$:
\begin{equation}\label{eq:regular}
1-G(x)=L(x)x^{-1/\gamma},\quad x>0,\; \gamma>0,
\end{equation}
where $L(\cdot)$ is a slowly varying function, that is,
\[\lim_{x\to\infty}L(tx)/L(x)=1,\quad t>0.\]
$L(\cdot)$ can be, for example, a constant or logarithmic function. We note that \eqref{eq:regular} describes a broad class of heavy-tailed distributions without imposing the rigid Pareto assumption.

An important contribution of this work is a novel analysis of the proposed algorithm that uses powerful results of the Extreme Value Theory (EVT)~--- a branch of probability that studies extreme events and properties of high order statistics in random samples. We refer to \cite{deHaan-Ferreira} for a comprehensive introduction to EVT.
Using EVT we can accurately predict the average fraction of correctly identified top-$k$ nodes and obtain the algorithm's complexity in terms of the number of nodes in $V$. We show that the complexity is sublinear if the in-degree distribution of the entities in $W$ is heavy tailed, which is usually the case in real networks.

The rest of the paper is organized as follows. In Section~\ref{sec:literature}, we give a short overview of related work. We formally describe our algorithm in Section~\ref{sec:algorithm}, then we introduce two performance measures in Section~\ref{sec:criteria}.  Section~\ref{sec:experiments} contains extensive experimental results that demonstrate the efficiency of our algorithm and compare it to baseline strategies. In Sections~\ref{sec:analysis}-\ref{sec:complexity} we present a detailed analysis of the algorithm and evaluate its optimal parameters with respect to the two performance measures. Section~\ref{sec:conclusion} concludes the paper.

\section{Related work}\label{sec:literature}

Over the last years data sets have become increasingly massive. For algorithms on such large data any complexity higher than linear (in dataset size) is unacceptable and even linear complexity may be too high. It is also well understood that an algorithm which runs in sublinear time cannot return an exact answer. In fact, such algorithms often use randomization, and then errors occur with positive probability.  Nevertheless, in practice, a rough but quick answer is often more valuable than the exact but computationally demanding solution. Therefore, sublinear time algorithms become increasingly important and many studies of such algorithms appeared in recent years (see, e.g., \cite{Sublinear1,Sublinear2,Sublinear3,Sublinear4}).

An essential assumption of this work is that the network structure is not available and has to be discovered using API requests. This setting is similar to on-line computations, where information is obtained and immediately processed while crawling the network graph (for instance the  World Wide Web). There is a large body of literature where such on-line algorithms are developed and analyzed. Many of these algorithms are developed for computing and updating the PageRank vector~\cite{Abiteboul,MonteCarloAvrachenkov,SublinearPageRank,Fagaras}. In particular, the algorithm recently proposed in \cite{SublinearPageRank} computes the PageRank vector in sublinear time. Furthermore, probabilistic Monte Carlo methods \cite{MonteCarloAvrachenkov,Bahmani2010PRMonteCarlo,Fagaras} allow to continuously update the PageRank as the structure of the Web changes.

Randomized algorithms are also used for discovering the structure of social networks. In \cite{Leskovec2006sampling} random walk methods are proposed to obtain a graph sample with similar properties as a whole graph. In \cite{Gjoka} an unbiased random walk, where each node is visited with equal probability,  is constructed in order to find the degree distribution on Facebook. Random walk based methods are also used to analyse Peer-to-Peer networks \cite{Massoulie2006}.
In \cite{Borgs2012pa_root} traceroute algorithms are proposed to find the root node and to approximate several other characteristics in a preferential attachment graph.

The problem of finding the most popular entities in large networks based only on the knowledge of a neighborhood of a current node has been analyzed in several papers. A random walk algorithm is suggested in \cite{Cooper2012high_degree} to quickly find the nodes with high degrees in a preferential attachment graph. In this case, transitions along undirected edges $x,y$ are proportional to $(d(x)d(y))^b$, where $d(x)$ is the degree of a vertex $x$ and $b > 0$ is some parameter.

In  \cite{Avrachenkov2012Top-k} a random walk with restart that uses only the information on the degree of a currently visited node was suggested for finding large degree nodes in undirected graphs. In \cite{Brautbar2010high_degree} a local algorithm for general networks, power law networks, and preferential attachment graphs is proposed for finding a node with degree, which is smaller than the maximal by a factor at most $c$.
Another crawling algorithm \cite{Kumar2008} is proposed to efficiently discover the correct set of web pages with largest incoming degrees in a fixed network and to track these pages over time when the network is changing. Note that the setting in \cite{Kumar2008} is different from ours in several aspects. For example, in our case we can use API to inquire the in-degree of any given item, while in the World Wide Web the information on in-links is not available, the crawler can only observe the in-links that come from the pages already crawled.

In Section~\ref{sec:baselines} we show that our algorithm outperforms the existing methods by a large margin.
Besides, several of the existing methods such as the ones in \cite{Avrachenkov2012Top-k} and \cite{Kumar2008} are designed specifically to discover the high degree nodes, and they cannot be easily adapted for other tasks, such as finding the most popular user categories or interest groups, while the algorithm proposed in this paper is simpler, much faster, and more generic.

To the best of our knowledge, this is the first work that presents and analyzes an efficient algorithm for retrieving the most popular entities under realistic API constraints.

\section{Algorithm description}\label{sec:algorithm}

Recall that we consider a bipartite graph $(V,W,E)$, where $V$ and $W$ are sets of entities and $(v,w)\in E$ represents a relation between the entities.

Let $n$ be the allowed number of requests to API. Our algorithm consists of two steps. We spend $n_1$ API requests on the first step and $n_2$ API requests on the second step, with $n_1+n_2=n$. See Algorithm~\ref{algo1} for the pseudocode.

\IncMargin{1em}
\begin{algorithm}
\SetKwInOut{Input}{input}\SetKwInOut{Output}{output}
\Input{Set of entities $V$ of size $N$, set of entities $W$ of size $M$, number of random nodes $n_1$ to select from $V$, number of candidate nodes $n_2$ from $W$}
\Output{Nodes $w_1, \dots w_{n_2}\in W$, their degrees $d_1, \dots, d_{n_2}$}
\BlankLine
\For{$w$ in $W$}{ $S[w] \leftarrow 0$\;}
\For{$i\leftarrow 1$ \KwTo $n_1$}{
  $v \leftarrow random(N)$\;
    \ForEach{$w$ in $OutNeighbors(v) \subset W$}{
    $S[w] \leftarrow S[w]+1$\;
  }}
$w_1, \dots, w_{n_2} \leftarrow Top\_n_2(S)$  // $S[w_1], \ldots, S[w_{n_2}]$ are the top $n_2$ maximum values in $S$\;
\For{$i\leftarrow 1$ \KwTo $n_2$}{ $d_i \leftarrow InDegree(w_i)$\;}
{\caption{Two-stage algorithm}
\label{algo1}}
\end{algorithm}\DecMargin{1em}

\textbf{First stage.} We start by sampling uniformly at random a set $A$ of $n_1$ nodes $v_1, \ldots, v_{n_1}\in V$. The nodes are sampled independently, so the same node may appear in $A$ more than once, in which case we regard each copy of this node as a different node. Note that multiplicities occur with a very small probability, approximately $1-e^{-{n_1^2}/{(2N)}}$. For each node in $A$ we record its out-neighbors in $W$. In practice, we bound the number of recorded out-links by the maximal number of IDs that can be retrieved within one API request, thus the first stage uses exactly $n_1$ API requests. For each $w\in W$ we identify $S[w]$, which is the number of nodes in $A$ that have a (recorded) edge to $w$.

\textbf{Second stage.} We use $n_2$ API requests to retrieve the actual in-degrees of the $n_2$ nodes with the highest values of $S[w]$. The idea is that the nodes with the largest in-degrees in $W$ are likely to be among the $n_2$ nodes with the largest $S[w]$. For example, if we are interested in the top-$k$ in-degree nodes in a directed graph, we hope to identify these nodes with high precision if $k$ is significantly smaller than $n_2$.

\section{Performance metrics}
\label{sec:criteria}

The main constraint of Algorithm~\ref{algo1} is the number of API requests we can use. Below we propose two performance metrics: the average fraction of correctly identified top-$k$ nodes and the first-error index.

We number the nodes in $W$ in the deceasing order of their in-degrees and denote the corresponding in-degrees by $F_1\ge F_2\ge \cdots\ge F_M$. We refer to $F_j$ as the $j$-th order statistic of the in-degrees in $W$. Further, let $S_j$ be the number of neighbors of a node $j$,  $1 \le j \le M$, among the $n_1$ randomly chosen  nodes in $V$, as described in Algorithm~\ref{algo1}. Finally, let $S_{i_1}\ge S_{i_2}\ge\ldots\ge S_{i_M}$ be the order statistics of $S_1,\ldots,S_M$. For example, $i_1$ is the node with the largest number of neighbors among $n_1$ randomly chosen nodes, although $i_1$ may not have the largest degree. Clearly, node $j$ is identified if it is in the set $\{i_1, i_2,\ldots, i_{n_2}\}$. We denote the corresponding probability by
\begin{equation}\label{eq:pj}
P_j(n_1):=\Prob(j\in \{i_1,\ldots,i_{n_2}\})\,.
\end{equation}

The first performance measure is the average fraction of correctly identified top-$k$ nodes. This is defined
in the same way as in~\cite{Avrachenkov2011Top-kPPR}:
\begin{align}
 \nonumber \E[\mbox{fraction of correctly identified top-$k$ entities}] &\\
 = \frac{1}{k}\sum_{j=1}^k P_j(n_1).&
 \label{eq:prediction}
\end{align}

The second performance measure is the first-error index, which is equal to $i$ if the top $(i-1)$ entities are identified correctly, but the top-$i$th entity is not identified. If all top-$n_2$ entities are identified correctly, we set the first-error index equal to $n_2+1$. Using the fact that for a discrete random variable $X$ with values $1,2,\ldots,K+1$ holds $\E(X)=\sum_{j=1}^{K+1}\Prob(X\ge j)$, we obtain the average first-error index as follows:
\begin{align}
\nonumber
\E[\mbox{1st-error index}]&=\sum_{j=1}^{n_2+1}\Prob(\mbox{1st-error index}\ge j)\\
&= \sum_{j=1}^{n_2+1}\prod_{l=1}^{j-1}P_l(n_1).
\label{eq:prediction1}
\end{align}

If the number $n$ of API requests is fixed, then the metrics \eqref{eq:prediction} and \eqref{eq:prediction1} involve an interesting trade-off between $n_1$ and $n_2$.
On the one hand, $n_1$ should be large enough so that the values $S_i$'s are sufficiently informative for filtering out important nodes. On the other hand, when $n_2$ is too small we expect a poor performance because the algorithm returns a top-$k$ list based mainly on the highest values of $S_i$'s, which have rather high random fluctuations. For example, on Figure~\ref{fig:fraction_correct}, when $n_2=k=100$, the algorithm returns the nodes $\{i_1,\ldots,i_{100}\}$, of which only 75\% belong to the true top-100. Hence we need to find the balance between $n_1$ and $n_2$. This is especially important when $n$ is not very large compared to $k$ (see Figure~\ref{fig:fraction_correct} with $n=1000$ and $k=250$).

\section{Experiments}\label{sec:experiments}

This section is organized as follows. First, we analyze the performance of our algorithm (most of the experiments are performed on the Twitter graph, but we also present some results on the CNR-2000 graph).
Then we compare our algorithm with baseline strategies on the Twitter graph and show that the algorithm proposed in this paper significantly outperforms existing approaches.
Finally, we demonstrate another application of our algorithm by identifying the most popular interest groups in the large online social network VKontakte.

All our experiments are reproducible: we use public APIs of online social networks and publicly available sample of a web graph.

\subsection{Performance of the proposed algorithm}

First, we show that our algorithm quickly finds the most popular users in Twitter. Formally, $V$ is a set of Twitter users, $W=V$, and $(v,w)\in E$ iff $v$ is a follower of $w$. Twitter is an example of a huge network with a very limited access to its structure.
Information on the Twitter graph can be obtained via Twitter public API.
The standard rate of requests to API is one per minute (see \textit{dev.twitter.com}).
Every vertex has an ID, which is an integer number starting from $12$.
The largest ID of a user is $\sim 1500$M (at the time when we performed the experiments).
Due to such ID assignment, a random user in Twitter can be easily chosen.
Some users in this range have been deleted, some are suspended,
and therefore errors occur when addressing the IDs of these pages. In our implementation we skip errors and assume that we do not spend resources on such nodes. The fraction of errors is approximately $30\%$.
In some online social networks the ID space can be very sparse and this makes problematic the execution
of uniform sampling in the first stage of our algorithm. In such situation we suggest to use random
walk based methods (e.g., Metropolis-Hastings random walk from \cite{Gjoka} or continuous-time random walk
from \cite{Massoulie2006}) that produce approximately uniform sampling after a burn-in period. To remove the
effect of correlation, one can use a combination of restart \cite{ART10} and thinning \cite{Avrachenkov2012Top-k,Gjoka}.

Given an ID of a user, a request to API can return one of the following: i) the number of followers (in-degree), ii) the number of followees (out-degree), or iii) at most 5000  IDs of followers or followees. If a user has more than 5000 followees, then all their IDs can be retrieved only by using several API requests. Instead, as described above, we record only the first 5000 of the followees and ignore the rest. This does not affect the performance of the algorithm because we record followees of randomly sampled users, and the fraction of Twitter users with more than 5000 followees is very small.

\begin{figure}
\centerline{\includegraphics[width = 0.5\textwidth]{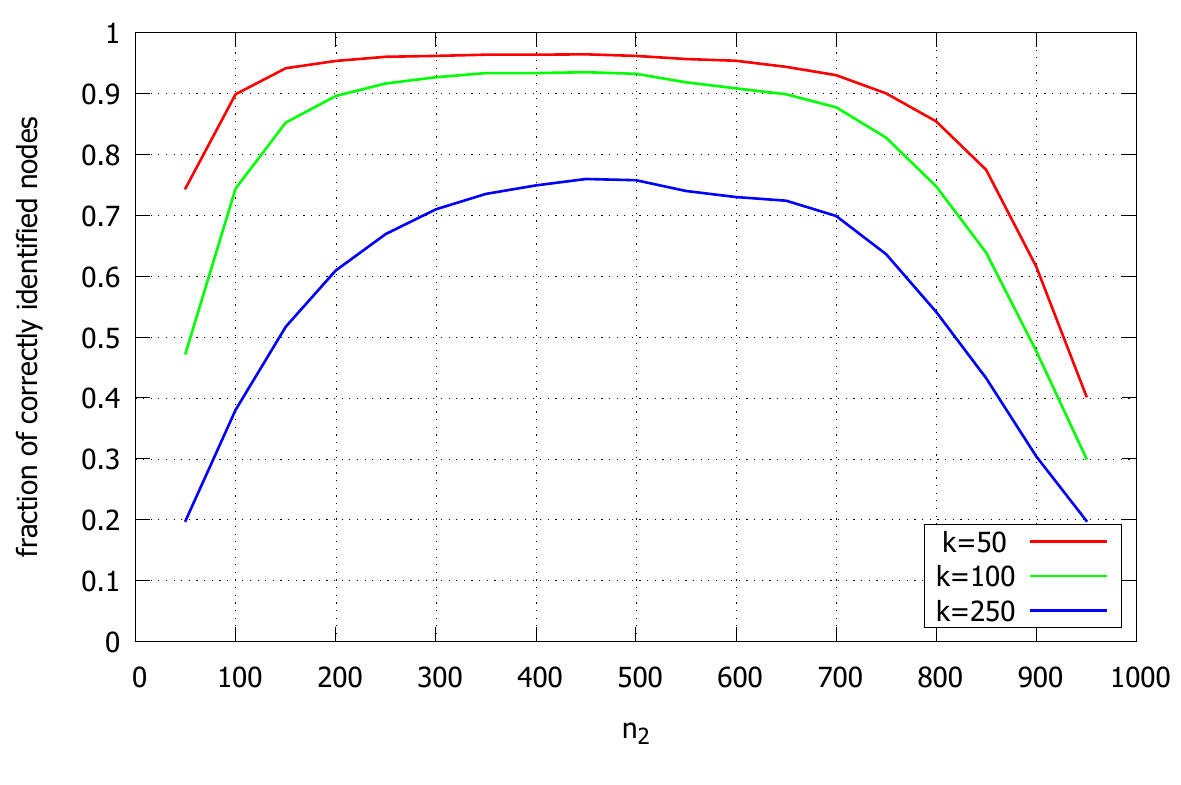}}
\caption{The fraction of correctly identified top-$k$ most followed Twitter users as a function of $n_2$, with $n=1000$.}
\label{fig:fraction_correct}
\end{figure}

In order to obtain the ground truth on the Twitter graph, we started with a top-1000 list from the publicly available source {\it twittercounter.com}. Next, we obtained a top-1000 list by running our algorithm with $n_1=n_2=20\,000$. We noticed that 1) our algorithm discovers all top-1000 users from {\it twittercounter.com}, 2) some top users identified by our algorithm are not presented in the top-1000 list on {\it twittercounter.com}. Then, we obtained the ground truth for top-1000 users by running our algorithm with ample number of API requests: $n_1 = n_2 = 500\,000$.

First we analyzed the fraction of correctly identified top-$k$ nodes  (see Equation~\eqref{eq:prediction}). Figure~\ref{fig:fraction_correct} shows the average fraction of correctly identified top-$k$ users for different $k$  over 100 experiments, as a function of $n_2$, when $n=1000$, which is very small compared to the total number of users. Remarkably we can find the top-50 users with very high precision. Note that, especially for small $k$, the algorithm has a high precision in a large range of parameters.

We also looked at the first-error index (see Equation~\eqref{eq:prediction1}), i.e., the position of the first error in the top list. Again, we averaged the results over 100 experiments. Results are shown on Figure~\ref{fig:first_mistake} (red line). Note that with only 1000 API requests we can (on average) correctly identify more than 50 users without any omission.

\begin{figure}
\centerline{\includegraphics[width = 0.5\textwidth]{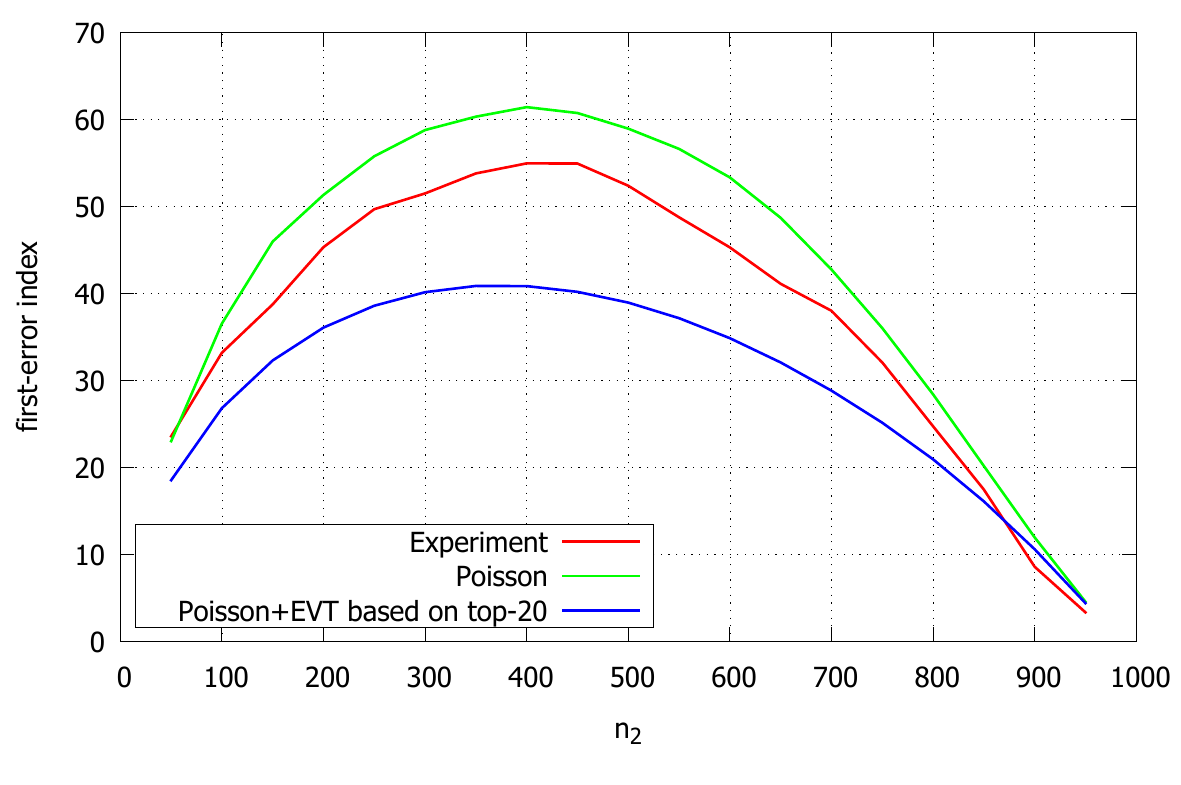}}
\caption{The first-error index as a function of $n_2$, with $n=1000$, on Twitter.}
\label{fig:first_mistake}
\end{figure}

\begin{figure}
\centerline{\includegraphics[width = 0.5\textwidth]{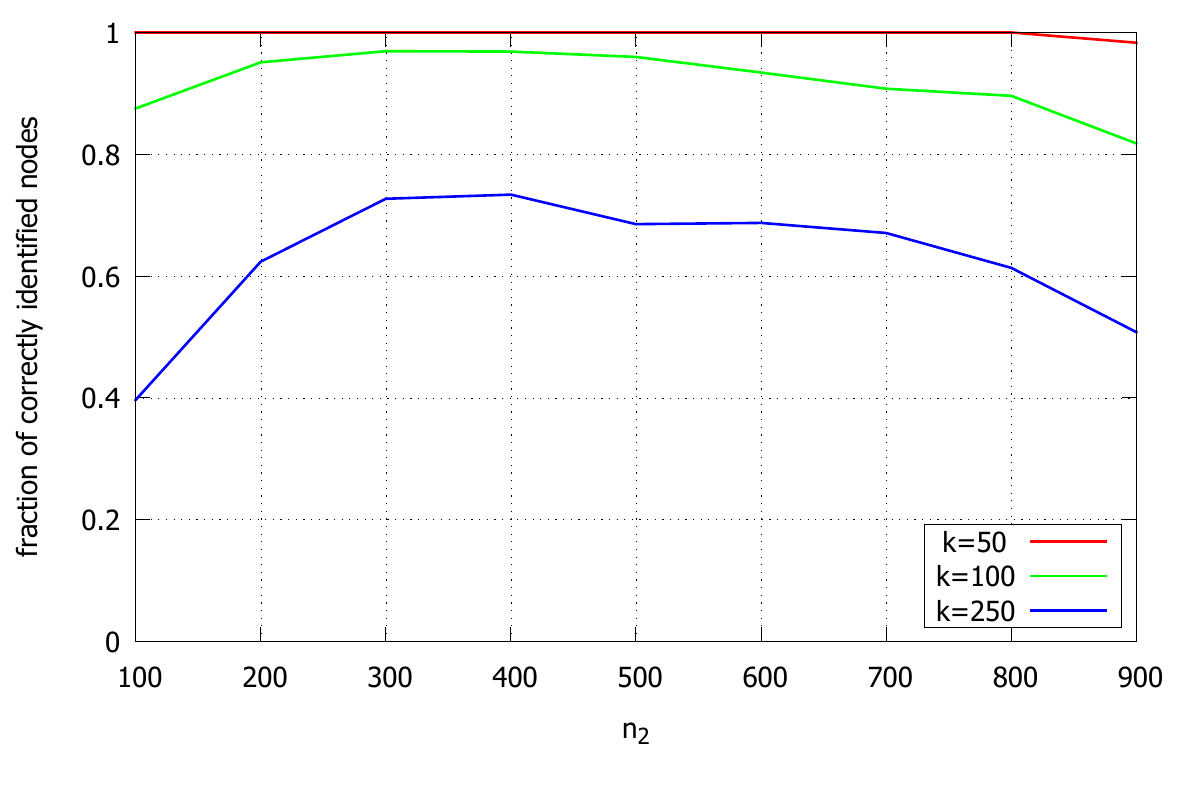}}
\caption{The fraction of correctly identified top-$k$ in-degree nodes in the CNR-2000 graph as a function of $n_2$, with $n=1000$.}
\label{fig:CNR}
\end{figure}

Although in this paper we mostly focus on the Twitter graph (since it is a huge network with a very limited rate of requests to API), we also demonstrated the performance of our algorithm on CNR-2000 graph ({\it law.di.unimi.it/webdata/cnr-2000}). This graph is a sample of the Italian CNR domain. It is much smaller and there are no difficulties in obtaining the ground truth here.
We get very similar results for this graph (see Figure~\ref{fig:CNR}).
Interestingly, the performance of the algorithm is almost insensitive to the network size: the algorithm performs similarly on the network with a billion nodes as on the network with half a million nodes.

\subsection{Comparison with baseline algorithms}\label{sec:baselines}

Literature suggests several solutions for the problem studied here. Not every solution is feasible in the setting of a large unknown realistic network. For example, random-walk-based algorithms that require the knowledge of the degrees of all neighbors of a currently visited node, such as the one in \cite{Cooper2012high_degree}, are not applicable. Indeed if we want to make a transition from a vertex of degree $d$, we need at least $d$ requests to decide where to go. So once the random walk hits a vertex of high degree, we may spend all the allowed resources on just one transition of the random walk.
In this section, we compare our algorithm with the algorithms suggested in \cite{Avrachenkov2012Top-k}, \cite{Brautbar2010high_degree}, and \cite{Kumar2008}. We start with the  description of these  algorithms.

\textbf{RandomWalk}~\cite{Avrachenkov2012Top-k}.

The algorithm in \cite{Avrachenkov2012Top-k} is a randomized algorithm for undirected graphs that
finds a top-$k$ list of nodes with largest degrees in sublinear time. This algorithm is based on a random walk with uniform jumps, described by the following transition probabilities \cite{ART10}:
\begin{equation}\label{eq:probrestart}
p_{ij} = \left\{ \begin{array}{ll}
\frac{\alpha/N+1}{d_i+\alpha}, & \mbox{if $i$ has a link to $j$},\\
\frac{\alpha/N}{d_i+\alpha}, & \mbox{if $i$ does not have a link to $j$},
\end{array}\right.
\end{equation}
where $N$ is the number of nodes in the graph and $d_i$ is the degree of node $i$. The parameter $\alpha$ controls how often the random walk makes an artificial jump. In \cite{Avrachenkov2012Top-k} it is suggested to take $\alpha$ equal to the average degree in order to maximize the number of independent samples, where the probability of sampling a node is proportional to its degree.
After $n$ steps of the random walk, the algorithm returns top-$k$ degree nodes from the set of all visited nodes. See Algorithm~\ref{algo2} for formal description.

\IncMargin{1em}
\begin{algorithm}
\SetKwInOut{Input}{input}\SetKwInOut{Output}{output}
\Input{Undirected graph $G$ with $N$ nodes, number of steps $n$, size of output list $k$, parameter $\alpha$}
\Output{Nodes $v_1, \dots v_{k}$, their degrees $d_1, \dots, d_{k}$}
\BlankLine
$v \leftarrow random(N)$\;
$A \leftarrow Neighbors(v)$\;
$D[v] \leftarrow size(A)$\;
\For{$i\leftarrow 2$ \KwTo $n$}{
    $r \xleftarrow{\text{sample}} U[0, 1]$\;
       \If{$r < \frac{D[v]}{D[v]+\alpha}$}{
         $v \leftarrow \text{ random from  }A$\;
       }
       \Else{
         $v \leftarrow random(N)$\;}
  $A \leftarrow Neighbors(v)$\;
  $D[v] \leftarrow size(A)$\; }
$v_1, \dots, v_{k} \leftarrow Top\_k(D)$  // $D[v_1], \ldots, D[v_{k}]$ are the top $k$ maximum values in $D$\;
\caption{RandomWalk}\label{algo2}
\end{algorithm}\DecMargin{1em}

Note that Algorithm~\ref{algo2} works only on undirected graphs. In our implementation on Twitter, all links in the Twitter graph are treated as undirected, and the algorithm returns the top-$k$ in-degree visited vertices. The idea behind this is that the random walk will often find users with large total number of followers plus followees, and since the number of followers of popular users is usually much larger than the number of followees, the most followed users will be found.
Another problem of Algorithm~\ref{algo2} in our experimental settings is that it needs to request IDs of all neighbors of a visited node in order to follow a randomly chosen link, while only limited number of IDs can be obtained per one API request (5000 in Twitter). For example, the random walk will quickly find a node with 30M followers, and we will need 6K requests to obtain IDs of all its neighbors. Therefore, an honest implementation of Algorithm~\ref{algo2} usually finds not more than one vertex from top-100. Thus, we have implemented two versions of this algorithm: strict and relaxed. One step of the strict version is one API request, one step of the relaxed version is one considered vertex. Relaxed algorithm runs much longer but shows better results. For both algorithms we took $\alpha = 100$, which is close to twice the average out-degree in Twitter.

\textbf{Crawl-Al and Crawl-GAI}~\cite{Kumar2008}.

We are given a directed graph $G$ with $N$ nodes.
At each step we consider one node and ask for its outgoing edges. At every step all nodes have their \textit{apparent in-degrees} $S_j$, $j=1,\ldots,N$: the number of discovered edges pointing to this node. In Crawl-Al the next node to consider is a random node, chosen with probability proportional to its apparent in-degree. In Crawl-GAI, the next node is the node with the highest apparent in-degree. After $n$ steps we get a list of nodes with largest apparent in-degrees. See Algorithm~\ref{Crawl-GAI} for the pseudocode of Crawl-GAI.

\IncMargin{1em}
\begin{algorithm}
\SetKwInOut{Input}{input}\SetKwInOut{Output}{output}
\Input{Directed graph $G$ with $N$ nodes, number of steps $n$, size of output list $k$}
\Output{Nodes $v_1, \dots v_{k}$}
\BlankLine
\For{$i\leftarrow 1$ \KwTo $N$}{ $S[i] \leftarrow 0$\;}
\For{$i\leftarrow 1$ \KwTo $n$}{
$v \leftarrow \mathrm{argmax}(S[i])$\;
$A \leftarrow OutNeighbors(v)$\;
\ForEach{$j$ in $A$}{
    $S[j] \leftarrow S[j]+1$\;
  }}
$v_1, \dots, v_{k} \leftarrow Top\_k(S)$  // $S[v_1], \ldots, S[v_{k}]$ are the top $k$ maximum values in $S$\;
\caption{Crawl-GAI}\label{Crawl-GAI}
\end{algorithm}\DecMargin{1em}

\textbf{HighestDegree}~\cite{Brautbar2010high_degree}.

A strategy which aims at finding the vertex with largest degree is suggested in \cite{Brautbar2010high_degree}. In our experimental setting with a limited number of API requests this algorithm can be presented as follows. While we have spare resources we choose random vertices one by one and then check the degrees of their neighbors. If the graph is directed, then we check the incoming degrees of out-neighbors of random vertices. See Algorithm~\ref{HighDegree} for the pseudocode of the directed version of this algorithm.


\IncMargin{1em}
\begin{algorithm}
\SetKwInOut{Input}{input}\SetKwInOut{Output}{output}
\Input{Directed graph $G$ with $N$ nodes, number of steps $n$, size of output list $k$}
\Output{Nodes $v_1, \dots v_{k}$, their degrees $d_1, \dots, d_{k}$}
\BlankLine
$s \leftarrow 0$\;
\For{$i\leftarrow 1$ \KwTo $n$}{\
  \If{$s=0$}{
        $v \leftarrow random(N)$\;
        $A \leftarrow OutNeighbors(v)$\;
        $s \leftarrow size(A)$\; }
  \Else{
        $D[A[s]] \leftarrow InDeg(A[s])$\;
        $s \leftarrow s-1$\;
        }
        }
$v_1, \dots, v_{k} \leftarrow Top\_k(D)$  // $D[v_1], \ldots, D[v_{k}]$ are the top $k$ maximum values in $D$\;
\caption{HighestDegree}\label{HighDegree}
\end{algorithm}\DecMargin{1em}

\medskip

The algorithms Crawl-AI, Crawl-GAI and HighestDegree find nodes of large in-degrees, but crawl only out-degrees that are usually much smaller. Yet these algorithms can potentially suffer from the API constraints, for example, when in-degrees and out-degrees are positively dependent so that large in-degree nodes tend have high number of out-links to be crawled. In order to avoid this problem on Twitter, we limit the number of considered out-neighbors by 5000 for these algorithms.

In the remainder of this section we compare our Algorithm~\ref{algo1} to the baselines on the Twitter follower graph.

The first set of results is presented in Table~\ref{tab:comparisonTwitter}, where we take the same budget (number of request to API) $n=1000$ for all tested algorithms to compare their performance. If the standard rate of requests to Twitter API (one per minute) is used, then 1000 requests can be made in 17 hours. For the algorithm suggested in this paper  we took $n_1 = 700$, $n_2=300$.

As it can be seen from Table~\ref{tab:comparisonTwitter}, Crawl-GAI algorithm, that always follows existing links, seems to get stuck in some densely connected cluster. Note that Crawl-AI, which uses randomization, shows much better results. Both Crawl-GAI and Crawl-AI base their results only on apparent in-degrees. The low precision indicates that due to randomness apparent in-degrees of highest in-degree nodes are often not high enough. Clearly, the weakness of these algorithms is that the actual degrees of the crawled nodes remain unknown. Algorithm~\ref{algo2}, based on a random walk with jumps, uses API requests to retrieve IDs of all neighbors of a visited node, but only uses these IDs to choose randomly the next node to visit. Thus, this algorithm very inefficiently spends the limited budget for API requests.
Finally, HighestDegree uses a large number of API requests to check in-degrees of all neighbors of random nodes, so it spends a lot of resources on unpopular entities.

Our Algorithm~\ref{algo1} greatly outperforms the baselines. The reason is that it  has several important advantages: 1) it is insensitive to correlations between degrees; 2) when we retrieve IDs of the neighbors of a random node (at the first stage of the algorithm), we increase their count of $S$, hence we do not lose any information; 3) sorting by $S[w]$ prevents the waste of resources on checking the degrees of unpopular nodes at the second stage;  4) the second stage of the algorithm returns the exact degrees of nodes, thus, to a large extent, we eliminate the randomness in the values of $S$.

\begin{table}
\begin{center}
\caption{Percentage of correctly identified nodes from top-100 in Twitter averaged over 30 experiments, $n = 1000$}\label{tab:comparisonTwitter}
\begin{tabular}{|l|c|c|}
\hline
Algorithm  & mean & standard deviation \\
\hline
\hline
Two-stage algorithm & 92.6 & 4.7 \\
\hline
RandomWalk (strict) & 0.43 & 0.63 \\
\hline
RandomWalk (relaxed) & 8.7 & 2.4 \\
\hline
Crawl-GAI & 4.1 & 5.9 \\
\hline
Crawl-AI & 23.9 &  20.2  \\
\hline
HighestDegree & 24.7 & 11.8 \\
\hline
\end{tabular}
\end{center}
\end{table}

On Figure~\ref{fig:dynamic_quality} we compare the average performance of our algorithm with the average performance of the baseline strategies for different values of $n$ (from 100 to 5000 API requests). For all values of $n$ our algorithm outperforms other strategies.

\begin{figure}
\centerline{\includegraphics[width = 0.5\textwidth]{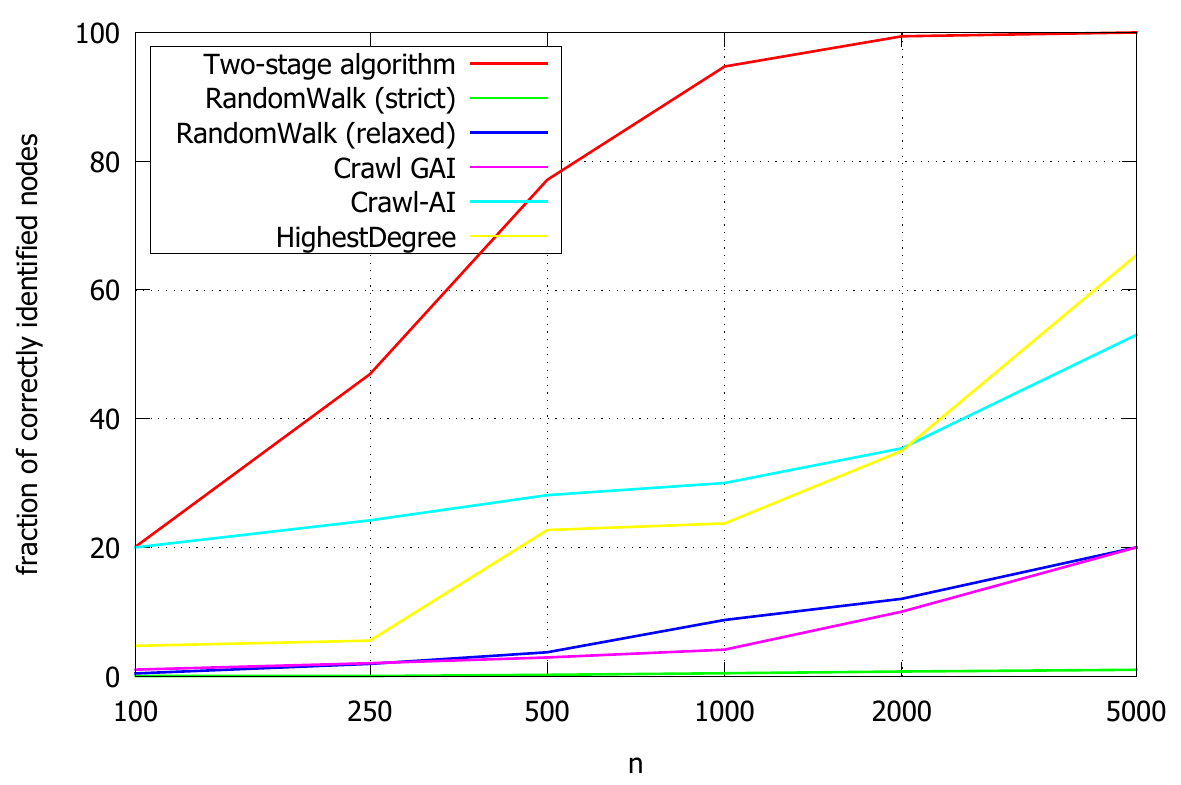}}
\caption{The fraction of correctly identified top-$100$ most followed Twitter users as a function of $n$ averaged over 10 experiments.}
\label{fig:dynamic_quality}
\end{figure}

\subsection{Finding the largest interest groups}\label{sec:groups}
\label{sec:groups}

In this section, we demonstrate another application of our algorithm: finding the largest interest groups in online social networks.
In some social networks there are millions of interest groups and crawling all of them may not be possible.
Using the algorithm proposed in this paper, the most popular groups may be discovered with a very small number of requests to API.
In this case, let $V$ be a set of users, $W$ be a set of interest groups, and $(v,w)\in E$ iff $v$ is a member of $w$.

Let us demonstrate that our algorithm allows to find the most popular interest groups in the large social network VKontakte with more than 200M registered users. As in the case of Twitter, information on the VKontakte graph can be obtained via API. Again, all users have IDs: integer numbers starting from $1$. Due to this ID assignment, a random user in this network can be easily chosen. In addition, all interest groups also have their own IDs.

We are interested in the following requests to API: i) given an ID of a user, return his or her interest groups, ii) given an ID of a group return its number of members.
If for some ID there is no user or a user decides to hide his or her list of groups, then an error occurs. The portion of such errors is again approximately $30\%$.

As before, first we used our algorithm with $n_1=n_2=50\,000$ in order to obtain the ground truth for the top-100 most popular groups (publicly available sources give the same top-100). Table~\ref{tab:groups} presents some statistics on the most popular groups.

\begin{table}[htb]
\begin{center}
\caption{The most popular groups for VKontakte}\label{tab:groups}
\begin{tabular}{|l|c|c|}
\hline
Rank  & Number of participants & Topic \\
\hline
\hline
1 & 4,35M & humor \\
\hline
2 & 4,10M  & humor \\
\hline
3 & 3,76M  & movies \\
\hline
4 & 3,69M & humor \\
\hline
5 & 3,59M  & humor \\
\hline
6 & 3,58M  & facts \\
\hline
7 & 3,36M  & cookery \\
\hline
8 & 3,31M & humor \\
\hline
9 & 3,14M  & humor \\
\hline
10 & 3,14M  & movies \\
\hline
\hline
100 & 1,65M  & success stories  \\
\hline
\end{tabular}
\end{center}
\end{table}

Then, we took $n_1 = 700$, $n_2 = 300$ and computed the fraction of correctly identified groups from top-100. Using only 1000 API requests, our algorithm identifies on average $73.2$ groups  from the  top-100 interest groups (averaged over 25 experiments). The standard deviation is $4.6$.

\section{Performance predictions}\label{sec:analysis}

In this section, we evaluate the performance of Algorithm~1 with respect to the metrics (\ref{eq:prediction}) and (\ref{eq:prediction1}) as a function of the algorithm's parameters $n_1$ and $n_2$.

Recall that without loss of generality the nodes in $W$ can be numbered $1,2,\ldots, M$ in the decreasing order of their in-degrees, $F_j$ is the unknown in-degree of a node $j$, and $S_j$ is the number of followers of a node $j$ among the randomly chosen $n_1$ nodes in $V$.

As prescribed by Algorithm~\ref{algo1}, we pick $n_1$ nodes in $V$ independently and uniformly at random with replacement. If we label all nodes from $V$ that have an edge to $j \in W$, then $S_j$ is exactly the number of labeled nodes in a random sample of $n_1$ nodes, so its distribution is $Binomial\Big(n_1,\frac{F_j}{N}\Big)$.
Hence we have
\begin{equation}
\label{eq:sj_exp}
\E(S_j)=n_1\,\frac{F_j}{N},\quad {\rm Var}(S_j)=n_1\,\frac{F_j}{N}\,\Big(1-\frac{F_j}{N}\Big).
\end{equation}

We are interested in predictions for the metrics (\ref{eq:prediction}) and (\ref{eq:prediction1}). These metrics are completely determined by the probabilities $P_j(n_1)$, $j=1,\ldots,k$, in (\ref{eq:pj}).
The expressions for $P_j(n_1)$, $j=1,\ldots,k$, can be written in a closed form, but they are computationally intractable because they involve the order statistics of $S_1,S_2,\ldots,S_M$. Moreover, these expressions depend on the unknown in-degrees $F_1,F_2,\ldots,F_M$.

We suggest two predictions for (\ref{eq:prediction}) and (\ref{eq:prediction1}). First, we give a {\it Poisson}  prediction that is based on the unrealistic assumption that the degrees $F_1,\ldots,F_{n_2}$ are known, and replaces the resulting expression for (\ref{eq:prediction}) and (\ref{eq:prediction1}) by an alternative expression, which is easy to compute. Next, we {suggest} an {\it Extreme Value Theory} (EVT) prediction that {does not require any preliminary knowledge of unknown degrees but uses the top-$m$ values of highest degrees obtained by the algorithm, where $m$ is much smaller than $k$.}

\subsection{Poisson predictions}\label{sec:Poisson}

First, for $j=1,\ldots,k$ we write
\begin{multline}
\label{eq:decompose}
P_j(n_1) =\\= \Prob(S_j>S_{i_{n_2}})+\Prob(S_j=S_{i_{n_2}},j\in\{i_1,\ldots,i_{n_2}\}).
\end{multline}
Note that if $[S_j>S_{i_{n_2}}]$ then the node $j$ will be selected by the algorithm, but if $[S_j=S_{i_{n_2}}]$, then this is not guaranteed and even unlikely. This observation is illustrated by the following example.

\begin{Example}\label{example:1} Consider the Twitter graph and take $n_1=700$, $n_2=300$. Then the average number of nodes $i$ with $S_i=1$ among the top-$l$ nodes is
\[
\sum_{i=1}^{l} \Prob(\mbox{$S_i=1$})
= \sum_{i=1}^{l} 700\,\frac{F_i}{10^9}\left(1-\frac{F_i}{10^9}\right)^{699},
\]
which is $223.3$ for $l=1000$, and it is $19.93$ for $l=n_2=300$. Hence, in this example, we usually see $[S_{i_{300}}= 1]$, however, only a small fraction of nodes with $[S_i=1]$ is selected (on a random basis) into the set $\{i_1,\ldots,i_{300}\}$.
\end{Example}

Motivated by the above example, we suggest to approximate $P_j(n_1)$ in (\ref{eq:decompose}) by its first term $\Prob(S_j>S_{i_{n_2}})$.

Next, we employ the fact that $S_{n_2}$ has the $n_2$-th highest average value among $S_1,\ldots, S_M$, and we suggest to use $S_{n_2}$ as a proxy for the order statistic $S_{i_{n_2}}$. However, we cannot replace $\Prob(S_j>S_{i_{n_2}})$ directly by $\Prob(S_j>S_{n_2})$ because the latter includes the case $[S_j>S_{n_2}=0]$, while with a reasonable choice of parameters it is unlikely to observe $[S_{i_{n_2}}=0]$.  This is not negligible as, e.g., in Example~\ref{example:1} we have $\Prob(S_{n_2}=0)\approx 0.06$. Hence, we propose to approximate $\Prob(S_j>S_{i_{n_2}})$ by $P(S_j>\max\{S_{{n_2}},1\})$, $j=1,\ldots,n_2$.

As the last simplification, we approximate the binomial random variables $S_j$'s by independent Poisson random variables. The Poisson approximation is justified because even for $j = 1, \ldots, k$ the value $F_j/N$ is small enough. For instance, in Example~\ref{example:1} we have $F_1/N\approx 0.04$, so $n_1 F_1/N$ is {$700\cdot 0.04=28$}.

Thus, summarizing the above considerations, we propose to replace $P_j(n_1)$ in (\ref{eq:prediction}) and (\ref{eq:prediction1}) by
\begin{equation}\label{eq:phat} \hat{P}_j(n_1)=\Prob(\hat{S}_j>\max\{\hat{S}_{{n_2}},1\}),\; j=1,\ldots,n_2,\end{equation}
where $\hat{S}_1,\ldots,\hat{S}_{n_2}$ are independent Poisson random variables with parameters $n_1F_1/N,\ldots, n_1F_{n_2}/N$. We call this method a Poisson prediction for (\ref{eq:prediction}) and (\ref{eq:prediction1}).

On Figures~\ref{fig:first_mistake}~and~\ref{fig:prediction}~the results of the Poisson prediction are shown by the green line. We see that these predictions closely follow the experimental results (red line).

\subsection{EVT predictions}\label{sec:EVT}
\label{eq:evt}

{Denote by $\hat F_1>\hat F_2>\cdots>\hat F_k$ the top-$k$ values obtained by the algorithm.}

Assume that the actual in-degrees in $W$ are randomly sampled from the distribution $G$ that satisfies (\ref{eq:regular}). Then $F_1>F_2>  \cdots>F_M$ are the order statistics of $G$.
{The EVT techniques allow to predict high quantiles of $G$ using the top values of $F_i$'s~\cite{Dekkers1989moment_estimator}. However, since the correct values of $F_i$'s are not known, we instead use the obtained top-$m$ values $\hat F_1, \hat F_2, \ldots, \hat F_m$, where $m$ is much smaller than $k$. This is justified for two reasons. First, given $F_j$, $j<k$, the estimate $\hat F_j$ converges to $F_j$ almost surely as $n_1\to\infty$, because, in the limit, the degrees can be ordered correctly using $S_i$'s only according to the strong law of large numbers.} {Second, when $m$ is small, the top-$m$ list can be found with high precision even when $n$ is very modest. For example, as we saw on Figure~\ref{fig:fraction_correct}, we find 50 the most followed Twitter users with very high precision using only 1000 API requests.}

{Our goal is to estimate $\hat{P}_j(n_1)$, $j=1,\ldots,k$, using only the values $\hat F_1,\ldots,\hat F_m$, $m<k$. To this end, we suggest to first estimate the value of $\gamma$ using the classical Hill's estimator $\hat{\gamma}$ \cite{Hill} based on the top-$m$ order statistics:}


\begin{equation}
\label{eq:hill}
\hat{\gamma}=\frac{1}{m-1}\sum_{i=1}^{m-1}(\log(\hat F_i)-\log(\hat F_m)).
\end{equation}

Next, we use the quantile estimator, given by formula (4.3) in \cite{Dekkers1989moment_estimator}, but we replace their two-moment estimator by the Hill's estimator in (\ref{eq:hill}). This is possible because both estimators are consistent (under slightly different conditions).  Under the assumption $\gamma>0$, we have the following estimator {$\hat f_j$} for the $(j-1)/M$-th quantile of $G$:
\begin{equation}
\label{eq:fj}
\hat f_j=\hat F_m\left(\frac{m}{j-1}\right)^{\hat{\gamma}},\qquad j>1, j<<M.
\end{equation}
We propose to use $\hat f_j$ as a prediction of the correct values $F_j$, $j=m+1,\ldots, n_2$.

Summarising the above, we suggest the following prediction procedure, which we call EVT prediction.
\begin{enumerate}
\item Use Algorithm~\ref{algo1} to find the top-$m$ list, $m << k$.
\item {Substitute the identified $m$ highest degrees $\hat F_1, \hat F_2,\ldots,\hat F_m$} in (\ref{eq:hill}) and (\ref{eq:fj}) in order to compute, respectively, $\hat\gamma$ and $\hat f_j$, $j=m+1,\ldots,n_2$.
\item Use the Poisson prediction (\ref{eq:phat}) substituting the values $F_1,\ldots, F_{n_2}$ by {$\hat F_1,\ldots, \hat F_m$, $\hat f_{m+1},\ldots, \hat f_{n_2}$.}
\end{enumerate}

On Figures~\ref{fig:prediction}~and~\ref{fig:first_mistake} the blue lines represent the EVT predictions, with $k=100$, $m=20$ and different values of $n_2$. For the average fraction of correctly identified nodes, depicted on Figure~\ref{fig:prediction}, we see that the EVT prediction is very close to the Poisson prediction and the  experimental results. The predictions for the first error index on Figure~\ref{fig:first_mistake} are less accurate but the shape of the curve and the optimal value of $n_2$ is captured correctly by both predictors. Note that the EVT prediction tends to underestimate the performance of the algorithm for a large range of parameters. This is because in Twitter the highest degrees are closer to each other than the order statistics of a regularly varying distribution would normally be, which results in an underestimation of $\gamma$ in (\ref{eq:hill}) if only a few top-degrees are used.

Note that the estimation \eqref{eq:fj} is inspired but not entirely justified by \cite{Dekkers1989moment_estimator} because the consistency of the proposed quantile estimator (\ref{eq:fj}) is only proved for $j<m$, while we want to use it for $j>m$. However, we see that this estimator agrees well with the data.

\begin{figure}
\centerline{\includegraphics[width = 0.5\textwidth]{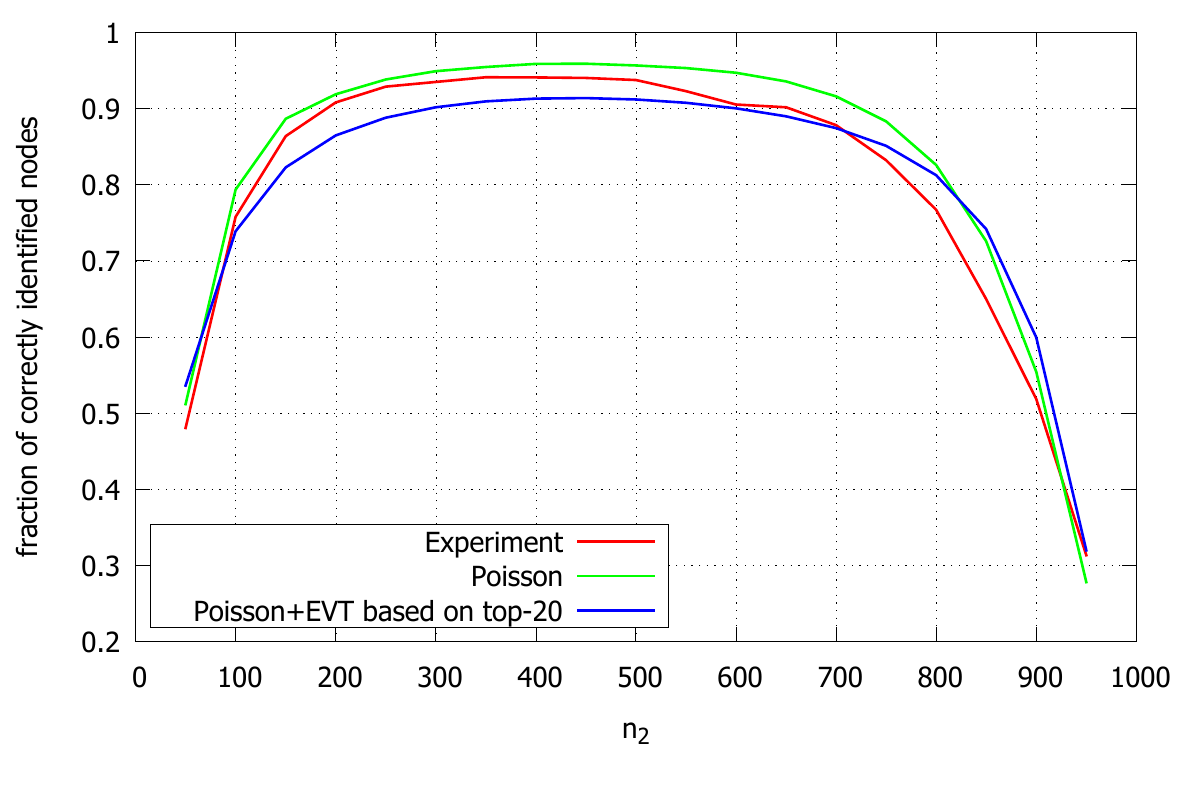}}
\caption{Fraction of correctly identified nodes out of top-100 most followed users in Twitter as a function of $n_2$, with $n=1000$.} 
\label{fig:prediction}
\end{figure}

\section{Optimal scaling for algorithm parameters}

In this section, our goal is to find the ratio $n_2$ to $n_1$ which maximizes the performance of Algorithm~\ref{algo1}. For simplicity, as a performance criterion we consider the expected fraction of correctly identified nodes from the top-$k$ list (see Equation~\eqref{eq:prediction}):
$$
\maximize_{n_1, n_2: n_1+n_2=n}\frac{1}{k}\sum_{j=1}^k P_j(n_1)\,.
$$

We start with  stating the optimal scaling for $n_1$. Let us consider the number of nodes with $S_j>0$ after the first stage of the algorithm. Assuming that the out-degrees of randomly chosen nodes in $V$ are independent, by the strong law of large numbers we have \[\limsup_{n_1\to\infty}\frac{1}{n_1}\sum_{j=1}^MI\{S_j>0\} \le \mu \quad \mbox{ with probability 1}, \] where $\mu$ is the average out-degree in $V$ and $I\{A\}$ is an indicator of the event $A$. Thus, there is no need to check more than $n_2=O(n_1)$ nodes on the second stage, which directly implies the next proposition.

\begin{prop}
\label{prop:n1} It is optimal to choose $n_1$ such that $n=O(n_1)$.
\end{prop}

As we noted before (see, e.g., Figure~\ref{fig:fraction_correct}), for small $k$ the algorithm has a high precision in a large range of parameters. However, for not too small values of $k$, the optimization becomes important. In particular, we want to maximize the value $P_k(n_1)$. We prove the following theorem.

\begin{theorem}
\label{th:n2}
Assume that $k=o(n)$ as $n \to \infty$, then the maximizer of the probability $P_k(n-n_2)$ is
$$
n_2= \left( 3 \gamma k^{\gamma} n \right)^{\frac{1}{\gamma+1}}\left(1 + \mbox{o}(1)\right),
$$
with $\gamma$ as in $(\ref{eq:regular})$.
\end{theorem}

\begin{proof}
It follows from Proposition~\ref{prop:n1} that $n_1 \to \infty$ as $n \to \infty$, so we can apply the following normal approximation
\begin{align}
\nonumber
P_k(n_1)& \approx P\left(N\left(\frac{n_1(F_k-F_{n_2})}{N},\frac{n_1(F_k+F_{n_2})}{N}\right)>0\right)\\
\label{eq:normal}
&= P\left(N(0,1) > -\sqrt{\frac{n_1}{N}}\frac{F_k - F_{n_2}}{\sqrt{F_k+F_{n_2}}}\right).
\end{align}
The validity of the normal approximation follows from the Berry-Esseen theorem. In order to maximize the above probability, we need to maximize $\sqrt{\frac{n_1}{N}}\frac{F_k - F_{n_2}}{\sqrt{F_k+F_{n_2}}}$. It follows from EVT that $F_k$ decays as $k^{-\gamma}$.
So, we can maximize
\begin{align}\label{eq:optim}
\frac{\sqrt{n-n_2}\left(k^{-\gamma} - n_2^{-\gamma}\right)}{\sqrt{k^{-\gamma}+n_2^{-\gamma}}}.
\end{align}
Now if $n_2=\mbox{O}(k)$, then $\sqrt{n-n_2}=\sqrt{n}(1+o(1))$ and the maximization of (\ref{eq:optim}) mainly depends on the remaining term in the product, which is an increasing function of $n_2$. This suggests that $n_2$ has to be chosen considerably greater than $k$. Also note that it is optimal to choose $n_2=o(n)$ since only in this case the main term in \eqref{eq:optim} amounts to $\sqrt{n}$. Hence, we proceed assuming the only interesting asymptotic regime where $k=o(n_2)$ and $n_2=o(n)$.
In this asymptotic regime, we can simplify (\ref{eq:optim}) as follows:
$$
\frac{\sqrt{n-n_2}\left(k^{-\gamma} - n_2^{-\gamma}\right)}{\sqrt{k^{-\gamma}+n_2^{-\gamma}}}=
$$
$$
\frac{1}{k^{\gamma/2}} \sqrt{n-n_2} \left( 1- \frac{3}{2} \left(\frac{k}{n_2}\right)^\gamma
+ \mbox{O}\left(\left(\frac{k}{n_2}\right)^{2\gamma}\right)\right).
$$
Next, we differentiate the function
$$
f(n_2):= \sqrt{n-n_2} \left( 1- \frac{3}{2} \left(\frac{k}{n_2}\right)^\gamma\right)
$$
and set the derivative to zero. This results in the following equation:
\begin{equation}\label{eq:n2eq}
\frac{1}{3\gamma k^\gamma} n_2^{\gamma+1} + n_2 - n = 0.
\end{equation}
Since $n_2 = \mbox{o}(n)$, then only the highest order term remains in (\ref{eq:n2eq}) and we immediately obtain the following approximation
$$
n_2= \left( 3 \gamma k^{\gamma} n \right)^{\frac{1}{\gamma+1}}\left(1 + \mbox{o}(1)\right).
$$
\end{proof}

\section{Sublinear complexity}\label{sec:complexity}

The normal approximation (\ref{eq:normal}) implies the following proposition.

\begin{prop}
\label{prop:1}
For large enough $n_1$, the inequality
\begin{equation}
\label{eq:prop1}
Z_k(n_1):=\sqrt{\frac{n_1}{N}}\frac{F_k - F_{n_2}}{\sqrt{F_k+F_{n_2}}} \ge z_{1-\varepsilon},
\end{equation}
where $z_{1-\varepsilon}$ is the $(1-\varepsilon)$-quantile of a standard normal distribution, guarantees that the mean fraction of top-$k$ nodes in $W$ identified by Algorithm~\ref{algo1} is at least $1-\varepsilon$.
\end{prop}

Using (\ref{eq:fj}), the estimated lower bound for $n_1$ in (\ref{eq:prop1}) is:
\begin{equation}
\label{eq:rough}
n_1 \ge  \frac{N z_{1-\varepsilon}^2 (k^{-\hat{\gamma}}+n_2^{-\hat{\gamma}})}{\hat F_m m^{\hat \gamma}(k^{-\hat{\gamma}}-n_2^{-\hat{\gamma}})^2}.
\end{equation}
In the case of the Twitter graph with $N=10^9$, $m=20$, $\hat F_{20}=18,825,829$, $k=100$, $n_2=300$, $z_{0.9}\approx 1.28$, $\hat\gamma=0.4510$, this will result in $n_1\ge 1302$, which is more pessimistic than $n_1=700$ but is sufficiently close to reality.  Note that Proposition~\ref{prop:1} is expected to provide a pessimistic estimator for $n_1$, since it uses the $k$-th highest degree, which is much smaller than, e.g., the first or the second highest degree.

We will now express the complexity of our algorithm in terms of $M$ and $N$, assuming that the degrees in $W$ follow a regularly varying distribution $G$ defined in (\ref{eq:regular}). In a special case, when our goal is  to find the highest in-degree nodes in a directed graph, we have $N=M$. If $M$ is, e.g., the number of interest groups, then it is natural to assume that $M$ scales with $N$ and $M\to\infty$ as $N\to\infty$. Our results specify the role of $N$, $M$, and $G$ in the complexity of Algorithm~\ref{algo1}.

From (\ref{eq:rough}) we can already anticipate that $n$ is of the order smaller than $N$ because $F_m$ grows with $M$. This argument is formalized in Theorem~\ref{th:complexity} below.

\begin{theorem}
\label{th:complexity}
Let the in-degrees of the entities in $W$ be independent realizations of a regularly varying distribution $G$ with exponent $1/\gamma$ as defined in (\ref{eq:regular}), and $F_1\ge F_2\ge\cdots\ge F_M$ be their order statistics.
Then for any fixed $\varepsilon,\delta>0$, Algorithm~\ref{algo1} finds the fraction $1-\varepsilon$ of top-$k$ nodes with probability $1-\delta$ in
\[n=\mbox{O}(N/a(M))\]
API requests, as $M,N\to\infty$, where $a(M)=l(M)M^\gamma$ and $l(\cdot)$ is some slowly varying function.
\end{theorem}

\begin{proof}
Let $a(\cdot)$ be a left-continuous inverse function of $1/(1-G(x))$. Then $a(\cdot)$ is a regularly varying function with index $\gamma$ (see, e.g., \cite{BinGolTeu89}), that is, $a(y)=l(y)y^{\gamma}$ for some slowly varying function $l(\cdot)$. Furthermore, repeating verbatim the proof of Theorem~2.1.1 in \cite{deHaan-Ferreira}, we obtain that for a fixed $m$
\[\left(\frac{F_1}{a(M)},\cdots,\frac{F_m}{a(M)}\right)\stackrel{d}{\to}\left(E_1^{-\gamma},\cdots,(E_1+\cdots+E_m)^{-\gamma}\right),\] where $E_i$ are independent exponential random variables with mean~1 and $\stackrel{d}{\to}$ denotes the convergence in distribution. Now for fixed $k$, choose $n_2$ as in Theorem~\ref{th:n2}. It follows that if $n_1=CN/a(M)$ for some constant $C>0$ then $Z_k(n_1)\stackrel{d}{\to}\sqrt{C(E_1+\cdots+E_k)^{-\gamma}}$ as $M,N\to\infty$. Hence, we can choose $C$, $M$, $N$ large enough so that $P(Z_k(n_1)>z_{1-\varepsilon})>1-\delta$. We conclude that $n_1=\mbox{O}(N/a(M))$ for fixed $k$, as $N,M\to\infty$. Together with Proposition~\ref{prop:n1}, this gives the result.
\end{proof}

In the case $M=N$, as in our experiments on Twitter, Theorem~\ref{th:complexity} states that the complexity of the algorithm is roughly of the order $N^{1-\gamma}$, which is much smaller than linear in realistic networks, where we often observe $\gamma\in(0.3,1)$~\cite{Newman}. The slowly varying term $l(N)$ does not have much effect since it grows slower than any power of $N$. In particular, if $G$ is a pure Pareto distribution, $1-G(x)=Cx^{-1/\gamma}$, $x\ge x_0$, then
$a(N)=C^{\gamma} N^{\gamma}$.

\section{Conclusion}\label{sec:conclusion}

In this paper, we proposed a randomized algorithm for quick detection of popular entities in large online social networks whose architecture has underlying directed graphs. Examples of social network entities are users, interest groups, user categories, etc. We analyzed the algorithm with respect to two performance criteria and compared it with several baseline methods. Our analysis demonstrates that the algorithm has sublinear complexity on networks with heavy-tailed in-degree distribution and that the performance of the algorithm is robust with respect to the values of its few parameters. Our algorithm significantly outperforms the baseline methods and has much wider applicability.

An important ingredient of our theoretical analysis is the substantial use of the extreme value theory. The extreme value theory is not so widely used in computer science and sociology but appears to be a very useful tool in the analysis of social networks. We feel that our work could provide a good motivation for wider applications of EVT in social network analysis. We validated our theoretical results on two very large online social networks by detecting the most popular users and interest groups.

We see several extensions of the present work. A top list of popular entities is just one type of properties of social networks. We expect that both our theoretical analysis, which is based on the extreme value theory, and our two-stage randomized algorithm can be extended to infer and to analyze other properties such as the power law index and the tail, network functions and network motifs, degree-degree correlations, etc. It would be very interesting and useful to develop quick and effective statistical tests to check for the network assortativity and the presence of heavy tails.

Since our approach requires very small number of API requests, we believe that it can be used for tracing network changes. Of course, we need formal and empirical justifications of the algorithm applicability for dynamic networks.

\section*{Acknowledgment}

This work is partially supported by the EU-FET Open grant NADINE (288956) and
CONGAS EU project FP7-ICT-2011-8-317672.

\bibliographystyle{abbrv}
\bibliography{references}

\end{document}